\documentclass[11pt,letterpaper]{amsart}
\usepackage[margin=1in]{geometry}
\usepackage{graphicx} % Required for inserting images
\newcommand{\hk}{Hegselmann--Krause~}
\usepackage{todonotes}

\title[$2R$-Conjecture for the H-K Model]{The $2R$-Conjecture for the Hegselmann--Krause Model: A Proof in Expectation and New Directions}
\author[Dey, Etesami \and Gopalan]{Partha S.~Dey \and S. Rasoul Etesami \and Aditya S.~Gopalan}
\date{\today}

\input{style}

\begin{document}

\begin{abstract}
\hk models are localized, distributed averaging dynamics on spatial data. A key aspect of these dynamics is that they lead to cluster formation, which has important applications in geographic information systems, dynamic clustering algorithms, opinion dynamics, and social networks. For these models, the key questions are whether a fixed point exists and, if so, characterizing it. In this work, we establish new results towards the ``$2R$-Conjecture’’ for the \hk model, for which no meaningful progress, or even any precise statement, has been made since its introduction in 2007. This conjecture relates to the structure of the fixed point when there are a large number of agents per unit space. We provide, among other results, a proof in expectation and a statement of a stronger result that is supported by simulation. The key methodological contribution is to consider the dynamics as an infinite-dimensional problem on the space of point processes, rather than on finitely many points. This enables us to leverage stationarity, shift invariance, and certain other symmetries to obtain the results. These techniques do not have finite-dimensional analogs.
\end{abstract}
\maketitle

%\todo[inline]{Many of the propositions do not have a written proof. We should put them in an appendix.}

\section{Introduction}
Consensus formation is a central theme in a variety of applications in social networks, distributed computer systems, and geographic information systems. Generally, in a consensus problem, different agents interact with the goal of reaching some common state. Many of these settings are \emph{spatial}, in the sense that agents only interact with other agents who are close to them; agents themselves may also move throughout time. In this case, consensus can be thought of as agents finding a common location in space. Thus, two broad themes emerge: (1) finding conditions under which consensus occurs, and (2) understanding the spatial structure of where consensus occurs. There are many practical applications for consensus formation. The following are just a few examples:

\begin{enumerate}[left=0pt, label=\emph{\alph*}.]
   \item  \textit{Robotic Rendezvous:} Consider a set of robots with limited communication capabilities. Each robot can communicate only with others within its communication radius, and collectively, they aim to rendezvous at the same location by iteratively interacting and updating their positions~\cite{bullo2009distributed}. 

    \item  \textit{Opinion Dynamics:} Opinion formation in social networks is a significant area of research that spans multiple disciplines, including psychology, economics, and political science. A common question across these fields is the extent to which one can predict the outcome of opinion formation among individuals engaged in complex interactions. One such interaction arises when social entities exhibit bounded confidence, meaning their opinions are influenced only by others whose views are sufficiently close to their own, within a specific confidence threshold~\cite{hegselmann2002opinion}.

    \item  \textit{Dynamic Clustering:} Dynamic clustering refers to the process of identifying and tracking clusters or groups in data that evolve over time. Unlike static clustering, where the structure is fixed, dynamic clustering accounts for temporal changes—such as the appearance, disappearance, merging, or splitting of clusters—as new data points are introduced or as existing data shifts. This is particularly important in applications such as social networks, mobile robotics, financial markets, and sensor networks, where the underlying data distribution is constantly changing~\cite{aggarwal2003framework}.
\end{enumerate}

Motivated by the applications mentioned above, this article primarily focuses on the \hk model for opinion dynamics, specifically from the second theme (\ie\ understanding the spatial structure of where opinion clusters occur), since the first theme has been widely studied for this model. Briefly, the \hk model is a discrete-time dynamical system defined as follows (see Section~\ref{sec:model} for a detailed description of the model and its assumptions).
\begin{itemize}[left=0pt]
    \item Agents are initially located on the real line $\mathbb{R}$ according to a simple stationary renewal process of intensity $\lambda$.
    \item Each agent interacts with all other agents within its $R$-ball (i.e., all agents located at a distance of at most $R$ from it). Specifically, at each time step, each agent moves to the average location of all agents (including itself) within its $R$-ball.
    \item A set of agents is said to reach consensus if they are co-located and there are no other agents within their common $R$-ball.
\end{itemize}

In the remainder of this article, we will refer to the agents as points; this terminology emphasizes that we treat the dynamics as dynamics on point processes.

Originally, the \hk dynamics were proposed for a \emph{restricted} deterministic setting with finitely many agents located on a compact interval~\cite{hegselmann2002opinion}. In this restricted setting, it is well-known that the dynamics reaches a fixed point after some finite time. Despite the apparent simplicity of the dynamics, little is known about the structure of the fixed point, even in the restricted setting. In this work, we extend the \hk model to a stochastic setting with infinitely many points in order to address the well-known $2R$-Conjecture~\cite{blondel20072r}. The statement (from~\cite{blondel20072r}) of this conjecture is given below: 
\begin{conjecture}
    In the restricted setting, as the number of points $N$ increases to infinity, the gaps between the clusters in the fixed point are \emph{approximately} $2R$.
\end{conjecture}

We note that no precise formulation of the approximately $2R$ gap between clusters has appeared in the existing literature ~\cite{blondel20072r, bhattacharyya2013convergence, wedin2022mathematics, bernardo2024bounded}. However, our extension of the \hk model to an initial condition consisting of a stationary renewal process on $\bR$, instead of on a compact interval, facilitates the understanding of the fixed point structure.
As a consequence of the finite-time consensus in the restricted setting, as time tends to infinity, there is a (random) weak limit of the extended dynamics. This weak limit consists of clusters of co-located points, with no other points located within the common $R$-ball of co-located points. We leverage the shift-invariance and stationarity properties of the process at all times $t$ to compute the expected distance between these clusters. We also note that, in a particular sense, the restriction to a compact interval represents an atypical instance of the \hk dynamics.
For more detailed comments, we refer the reader to Remark~\ref{rem:interval-atypical}.

\subsection{Contributions}
Our contributions are as follows. 
Noting that taking $N \to \infty$ in the restricted setting is analogous to sending $\lambda \to \infty$ in the extended setting, we show that the expected gap between adjacent clusters in the weak limit converges to $2R$ as $\lambda \to \infty$. We refer to this result as the \emph{Weak $2R$-Conjecture}.
We also state, as a stronger result, the \emph{Strong $2R$-Conjecture}, which is supported by simulation evidence.
Roughly speaking, this conjecture is that the $\lambda\to\infty$ limit of the weak limit fixed points (as time $ t\to\infty$) is the stationary version of $2R\cdot \bZ$.

\subsection{Related Literature}
There is a large literature on the \hk model, for which we refer the reader to the recent survey article~\cite{bernardo2024bounded} and to the references therein. However, we specifically mention the following articles that place significant focus on the $2R$-conjecture for approximations of the \hk model~\cite{blondel20072r, blondel2010continuous, wang2017noisy}.
This paper is the first work to state and prove a result for the $2R$-Conjecture for the original \hk dynamics. 

We also mention some recent work in the applied probability literature related to spatial dynamics on point processes.
Resolving our Strong $2R$-Conjecture remains an open problem, and tools to prove related convergences are yet to be developed for any of these models~\cite{dragovic2020poisson, khaniha2025hierarchical, angel2023tale}.

\subsection{Organization}
The rest of this article is organized as follows.
In Section~\ref{sec:notation}, we provide some notations and preliminaries that are used throughout the paper. The model and main results are given in Section~\ref{sec:model}.
We prove the Weak $2R$-Conjecture in Section~\ref{sec:weak-2R-proof}.
Simulation evidence for the Strong $2R$-Conjecture is in Section~\ref{sec:strong-2R-evidence}.
Finally, we discuss some directions for future work in Section~\ref{sec:future-work}.

\section{Preliminaries and Notations}\label{sec:notation}

In this section, we first provide some definitions and preliminary results on point processes and geometric random graphs. We then give some notations that will be used throughout the paper.

\subsection{Preliminaries}
We will restrict any definitions to $\bR$ instead of the more general $\bR^n$ since that is the setting considered in this paper.
All results in this section are standard and will be stated without proof.

For this section only, let $(\Omega, \cF, \pr)$ be a probability space on the non-decreasing elements of $\bR^\bZ$.

\begin{definition}
    A \emph{point process} on $\bR$ is a random countable collection of points on $\bR$, which is, without loss of generality, a non-decreasing element of $\bR^\bZ$.
    A point process $\Xi$ is \emph{simple} if, for all $x \in \bR$, we have $\pr(\text{more than one point of $\Xi$ is at $x$}) = 0$.
    A point process $\Xi$ is \emph{stationary} if for any two sets $S_1, S_2 \subset \bR$ of equal and finite Lebesgue measure, the distributions of the number of points in $S_1$ and $S_2$ are equal.
\end{definition}

For a point process $\Xi$ on $\bR$, we will take as the $0$-indexed point the point of least positive location.
We will index points of greater location successively by $1, 2, 3, \ldots$, and points of lesser location successively by $-1, -2, -3, \ldots$.

\begin{definition}
    A simple point process $\Xi$ on $\bR$ is \emph{renewal} if the gaps between successive points are independent and identically distributed.
    A renewal process is \emph{Poisson} if the gaps between successive points are independent and identically distributed as $\mathrm{exp}(\lambda)$ for some fixed $\lambda > 0$.
\end{definition}

Let $N(t)$ be the number of points of a stationary point process $\Xi$ in $(0, t]$, for $0 \neq t \in \bR$.
\begin{definition}
    The \emph{intensity} of $\Xi$ is $\lambda := \lim_{t \to \pm\infty}\frac{\E N(t)}{t}$, provided it exists.
\end{definition}

In general, the intensity of a point process could be infinite or random.
In this paper, the intensity of any point process that we consider exists, is finite, and is deterministic.
It will be convenient to work with an alternative characterization of the intensity of a stationary point process. Proofs of the following propositions are standard; we omit them for brevity.

\begin{proposition}\label{prop:int}
    If the (finite and deterministic) intensity $\lambda$ of a stationary point process $\Xi$ exists, then $\lim_{t \to\pm \infty}\frac{N(t)}{t} = \lambda$ almost surely.
    \label{prop:intensity-empirical}
\end{proposition}

We will also use the following result, which relates intensity to the expected gap between consecutive points.
\begin{proposition}\label{prop:gap}
    A simple stationary point process $\Xi$ has intensity $\lambda$ iff the expected gap between consecutive points is $\frac{1}{\lambda}$.
    \label{prop:intensity-gap}
\end{proposition}

It will sometimes be helpful to think of point processes in the context of \emph{geometric random graphs}.
A geometric random graph can be constructed from any (finite or countable) set $S$ of points in $\bR$ by adding an edge between points $x,y\in S$ whenever $|y-x| \leq r$ for some fixed $r$.

\begin{definition}
    Let $\Xi$ be a stationary point process, denote its $x$-indexed point by $\Xi_x$. For $-\infty \leq x \leq y \leq \infty$, the interval $[x, y] \subseteq \bZ$ is a \emph{connected component} for $\Xi$ if all of the following hold:
    \begin{enumerate}[label={\arabic*.}]
        \item If $-\infty \leq x < z \leq y \leq \infty$, $\Xi_z - \Xi_{z-1} \leq 1$,
        \item If $x > -\infty$, $\Xi_{x} - \Xi_{x-1} > 1$,
        \item If $y < \infty$, $\Xi_{y+1} - \Xi_{y} > 1$.
    \end{enumerate}
\end{definition}

\begin{definition}
    The \emph{induced graph} w.r.t.~$\Xi$ for a connected component $[x, y] \subseteq \bZ$ is the graph with vertex set $\{\Xi_{v} : v \in [x, y]\}$ and with an undirected edge $(\Xi_{x}, \Xi_{y})$ whenever $|\Xi_{y}  - \Xi_{x}| \leq 1$.
\end{definition}

Throughout, we will refer to connected components and their induced graphs interchangeably.
We will also refer to points of $\Xi$ and vertices of the induced graph interchangeably.

\subsection{Notations}
In this work, we often look at the evolution of a point process under certain dynamics. When $\Xi_{\lambda}^{(0)}$ is a simple renewal process of intensity $\lambda$, we will denote by $\Xi_{\lambda}^{(t)}$ the resultant point process under the dynamics at time $t$. We abuse notation and write $\Xi_{\lambda}^{(t)}(\omega)$ when it helps to explicitly treat a point process as a random variable.
Where it is clear, we will sometimes suppress the dependence on $\omega$. If we also need to refer to the $i$-indexed point, we will write $\Xi_{\lambda, i}^{(t)}$ or $\Xi_{\lambda,i}^{(t)}(\omega)$.
When the intensity $\lambda$ is not needed, we will use $\Xi_i$ to denote the $i$-indexed point of $\Xi$.
We will refer to the $i$-indexed point as the $i$-th point or point $i$. Finally, for a set $S$, we denote by $\#S$ the number of elements of $S$. For $R \geq 0$, we denote by $B_R(x)$ the closed $R$--ball centered at $x$, i.e., $B_R(x)=\{y\in \mathbb{R}: |x-y|\leq R\}$.

\section{Model, $2R$-Conjecture, and Main Results}
\label{sec:model}
\subsection{Model}
We consider the following model, which is an extension of the model introduced by Hegselmann and Krause in 2002~\cite{hegselmann2002opinion}.
Let $\Xi_{\lambda}^{(0)}$ be a simple renewal process on $\bR$ with intensity $\lambda$.
Throughout, we assume that the distribution of the gap between successive points of $\Xi_{\lambda}^{(0)}$ has a density.
The dynamics are in discrete time as follows.
At each time $t \geq 1$, the location of the $i$-th point $\Xi_{\lambda,i}^{(t)}$ is the arithmetic average of all points within its $R$--ball at time $t-1$, \textit{i.e.},
$$\Xi_{\lambda,i}^{(t)} 
:= \frac{\sum_{j \in \bZ}\Xi_{\lambda,j}^{(t-1)} \cdot \ind_{ \Xi_{\lambda,j}^{(t-1)} \in B_R(\Xi_{\lambda,i}^{(t-1)})}}{\#\big\{j \in \bZ: \Xi_{\lambda,j}^{(t-1)} \in B_R(\Xi_{\lambda,i}^{(t-1)})\big\}},$$
where $\ind_{(\cdot)}$ is the indicator function. \emph{Without loss of generality, we will take $R = 1$ in this paper}.

We note that the original model by Hegselmann and Krause considers only a finite number of initial points restricted to a compact interval of $\bR$; this is the setting where many of the analyses have been performed.
It is well--known that in this setting, the dynamics reach a fixed point almost surely, within finite time (see the survey article~\cite{bernardo2024bounded} and the references therein).

We take as our initial condition a point process of positive intensity; otherwise, our dynamics are identical to those of the \hk model.
The techniques we use in this paper do not have analogues in the restricted setting of a compact interval.

\subsection{$2R$-Conjecture}
Since it has been established that each point stops moving a.s.~after a finite time, the next natural question is determining the locations at which points remain fixed.
The fixed point consists of \emph{clusters} of co-located points.
It is easy to see that the fixed points of the \hk dynamics are not unique: for any $\varepsilon > 0$, the simple point process $(1 + \varepsilon)\cdot \bZ$ is a fixed point.
Similarly, any renewal process whose gap distribution $G$ satisfies $\pr(G \leq 1) = 0$ is also a fixed point.

The $2R$-Conjecture, introduced by Blondel, Hendrickx, and Tsitsiklis~\cite{blondel20072r}, states that for large $\lambda$, the gaps between clusters of $\Xi_{\lambda}^{(\infty)}$ are \emph{approximately} equal to $2$; although the precise sense of the word ``approximately'' is not specified.
As of recent works~\cite{bhattacharyya2013convergence, bernardo2024bounded, wedin2022mathematics}, this conjecture has been neither stated precisely nor resolved in the literature.

\subsection{Main Results}
In this paper, we show the following.
We note that even the precise statements of the results are new and have not appeared before in the literature despite significant general interest in the \hk model and the $2R$-conjecture specifically~\cite{bernardo2024bounded}.
Throughout this paper, we argue on $(\Omega, \cF, \pr)$, which is a probability space including point processes on $\bR$ of all intensities $\lambda > 0$.

\begin{theorem}
    The weak limit $\Xi_{\lambda}^{(\infty)} := \lim_{t \to \infty}\Xi_{\lambda}^{(t)}$ exists $\pr$-a.s..
    \label{thm:limit-existence}
\end{theorem}

\begin{theorem}[Weak $2R$-Conjecture]
    For all $\lambda > 0$, the point process $\Xi_{\lambda}^{(\infty)}$ is stationary and has intensity $\frac{\lambda}{2\lambda + 1}$ $\pr$-a.s..
    In particular, as $\lambda \to \infty$, the expectation of the gap between clusters of $\Xi_{\lambda}^{(\infty)}$ converges to $2$.
    \label{thm:weak2R}
\end{theorem}

The following definition gives us a condition on the initial conditions $(\Xi_\lambda^{(0)})_{\lambda \in \bR}$ under which we can state the Strong $2R$-conjecture.
\begin{definition}
    The sequence $(\Xi_\lambda^{(0)})_{\lambda \in \bR}$ of initial conditions \emph{satisfies the functional strong law of large numbers} if the following holds.
    Let $F_\lambda$ be a non-decreasing right-continuous function with left limits such that $F_\lambda(0) = 0$, and for $y > x$, $F_\lambda(y) - F_\lambda(x)$ is the number of points of $\Xi_\lambda^{(0)}$ in $(x, y]$.
    Then as $\lambda \to \infty$, $$\frac{1}{\lambda}F_{\lambda} \xrightarrow{u.o.c.} \mathrm{Id}_\bR,$$ 
    where $\mathrm{Id}_\bR$ is the identity function on $\bR$ and \textit{u.o.c.} denotes uniform convergence on compact sets of $\bR$.
\end{definition}

\begin{conjecture}[Strong $2R$-Conjecture]
    Let $(\Xi_{\lambda}^{(0)})_{\lambda \in \bR}$ be a sequence of simple stationary renewal processes that satisfy the functional strong law of large numbers. Then the following hold:
    \begin{enumerate}
        \item The joint locations of the clusters of $\lim_{\lambda \to \infty}\lim_{t \to \infty} \Xi_{\lambda}^{(t)}$ are distributed as $\theta + 2\bZ$, where $\theta \sim \mathrm{Unif}(-1, 1)$.
        \item If $X_{\lambda}(i)$ is the number of points in the $i$-th cluster of $\lim_{t \to \infty}\Xi_{\lambda}^{(t)}$, then as $\lambda \to \infty$, the joint convergence holds: $(\frac{1}{\lambda}X_{\lambda}(i))_{i \in \bZ} \xrightarrow{a.s.} (\ldots, 2, 2, 2, \ldots)$.
    \end{enumerate}
\end{conjecture}

We provide evidence for our statement of the strong $2R$-Conjecture in Section~\ref{sec:strong-2R-evidence}.
The assumption that the initial conditions satisfy the functional strong law of large numbers is mild and is satisfied, \textit{e.g.}, when $\Xi_{\lambda}^{(0)}$ is a Poisson point process of intensity $\lambda$.

\section{Proofs of the Theorems}
\label{sec:weak-2R-proof}

\subsection{Proof of Theorem~\ref{thm:limit-existence}}
Let $p_{\lambda}^1$ be the probability that successive points of $\Xi_{\lambda}^{(0)}$ are at a distance of at most $1$.
Notice that at time $0$, the size of a connected component is distributed as $\mathrm{Geom}(1 - p_{\lambda}^{1})$, and thus the connected components are a.s. finite at time $0$.
It is well-known that the \hk dynamics within each connected component terminate in finite time, from which the result follows.

%It is a straightforward consequence that the simplified point process $\widehat{\Xi}_{\lambda}^{(\infty)}$ also exists.

\begin{remark}\label{rem:interval-atypical}
    The connected component argument used to show the existence of the weak limit $\Xi_{\lambda}^{(\infty)}$ serves as an important motivator for the extension of the \hk models to point processes.
    Even the imprecise notions of the $2R$ conjectures are in the setting where the number of points in a compact interval is sent to infinity: this is analogous to sending the intensity of a point process to infinity.
    Using, for exposition, the assumption that $\Xi_{\lambda}^{(0)}$ is a Poisson point process of intensity $\lambda$, the size of the connected components of the initial point process are distributed as $\mathrm{Geom}(e^{-\lambda})$.
    While for any $K < \infty$, there are infinitely many connected components that are contained in a compact interval of length $K$, these are vanishingly rare as $\lambda \to \infty$.
    
\end{remark}

\subsection{Basic Properties}
Denote by $\cC_{\lambda}^{(t)}$ the connected component containing the $0$-th vertex at time $t$, when the intensity of the initial condition is $\lambda$. Also, denote the number of points in the connected component $\cC_{\lambda}^{(t)}$ by $\#\cC_{\lambda}^{(t)}$. We will treat $\cC_{\lambda}^{(t)}(\omega)$ as a random geometric graph rooted at the $0$-th vertex.
The following are easily seen to hold:

\begin{proposition}
    For all $\lambda > 0$, $\#\cC_{\lambda}^{(0)} \sim \mathrm{Geom}_1(p)$, where $p := \pr(\Xi_{\lambda,1}^{(0)} - \Xi_{\lambda,0}^{(0)} > 1)$ is the probability that the gap between consecutive points in $\Xi_{\lambda}^{(0)}$ exceeds $1$.
    Furthermore, the induced geometric random graphs for connected components of $\Xi_{\lambda}^{(0)}$ are i.i.d. $\pr$-a.s. finite geometric random graphs distributed according to an $\cF$-measurable probability measure $\pr^{(0)}$ on finite geometric random graphs on $\bR$.
    \label{prop:iid-comps-0}
\end{proposition}

\begin{proposition}
    The induced graph $\Xi_{\lambda}^{(t)}$ consists of i.i.d. $\pr$-a.s. finite random graphs (although each component of $\Xi_{\lambda}^{(t)}$ is not necessary independent of the next). 
    In particular, the $\cF$-measurable law $\pr^{(t)}$ of $\cC_{\lambda}^{(t)}(\omega)$ is well-defined on the space of finite geometric random graphs on $\bR$.
    \label{prop:law_C_lambda_t}
\end{proposition}

Note that the measures $(\pr^{(t)})_{t \geq 0}$ are not identical over $t$ but they can all be defined over $\cF$.

\begin{proposition}
    For all $\lambda > 0$, and for all $t \geq 0$, $\#\cC_{\lambda}^{(t)}(\omega) \geq \#\cC_{\lambda}^{(t+1)}(\omega)$, $\pr(\omega)$ a.s..
    \label{prop:comp-size-monotonicity}
\end{proposition}

\begin{proposition}
    Conditionally on the unlabeled connectivity structure of $\cC_{\lambda}^{(t)}$, the $0$-th vertex is uniformly located within $\cC_{\lambda}^{(t)}$. 
    \label{prop:0pt-uniform}
\end{proposition}
\begin{proof}
    Suppose that $\cC_{\lambda}^{(t)}$ has $n$ points labeled in $[-i, j]$.
    By shift invariance, the law of $\Xi_{\lambda}^{(0)}(\omega)$ would be unchanged if we uniformly shift indices by any integer in $[-i, j]$. 
    Such a re-labeling of the points in $\Xi_{\lambda}^{(0)}$ results in the same induced random geometric graph with appropriately re-labeled vertices and the root vertex $0$ in a different location.
    The result follows.
\end{proof}

While it is not mentioned explicitly in its statement, the following corollary relies crucially on the fact that the $0$ point is uniformly located in its connected component at any time $0 \leq t \leq \infty$.
\begin{corollary}
    Let 
    \begin{align*}
        \Big(h_{\lambda}^{(t)}(u, v)\Big)(\omega) := &\#\Big\{\text{points of $\Xi_{\lambda}^{(t)}(\omega)$ in $B_1\big(\Xi_{\lambda, u}^{(t)}(\omega)\big)$}
       \text{ and also in $B_1\big(\Xi_{\lambda, v}^{(t+1)}(\omega)\big)$} \Big\}.
    \end{align*}

    Then 
    \begin{equation}
        \E_{\cC_{\lambda}^{(t)}(\omega)}\left[\sum_{v \in \cC_{\lambda}^{(t)}(\omega)} \Big(h_{\lambda}^{(t)}(v, 0)\Big)(\omega)\right] = \E_{\cC_{\lambda}^{(t)}(\omega)}\left[\sum_{v \in \cC_{\lambda}^{(t)}(\omega)} \Big(h_{\lambda}^{(t)}(0, v)\Big)(\omega)\right].
        \label{eq:MTP}
    \end{equation}
    \label{prop:MTP}
\end{corollary}
\begin{proof}
    This follows from a similar shift invariance argument as in the proof of Proposition~\ref{prop:0pt-uniform}.
    %Fix any two vertices $u, v \in \cC_{\lambda}^{(t)}$ such that $\Big(h_{\lambda}^{(t)}(u, v)\Big)(\omega) > 0$.
    It is clear from Proposition~\ref{prop:0pt-uniform} that the result holds conditionally on the locations of points (but not the labels) in $\cC_{\lambda}^{(t)}$.
    Since Proposition~\ref{prop:law_C_lambda_t} establishes that the law $\pr^{(t)}$ of $\cC_{\lambda}^{(t)}$ is well-defined, the expectation is also well-defined and the result follows.
\end{proof}

\begin{remark}
    Equation~(\ref{eq:MTP}) is an example of the \emph{Mass Transport Principle}, which has enjoyed successful application to various problems regarding random graphs and point processes.
    For a more detailed discussion, see~\cite{aldous2007processes, blaszczyszyn2017lecture} and the references therein.
    We also refer the reader to~\cite{gwynne2018invariance} for a version of the Mass Transport Principle, which more closely resembles its use in this paper.
\end{remark}

\begin{corollary}\label{cor:int}
    For all $\lambda > 0$, $t \geq 0$, $\E[\#\cC_{\lambda}^{(t+1)}(\omega)] = \E[\#\cC_{\lambda}^{(t)}(\omega)].$ 
    In particular, for all $\lambda > 0$, $t \geq 0$, we have $\E[\#\cC_{\lambda}^{(t)}] = 2\lambda + 1$.
\end{corollary}
\begin{proof}
    The first statement follows immediately from Corollary~\ref{prop:MTP}.
    The second statement follows from the fact that $\Xi_{\lambda}^{(0)}(\omega)$ is a simple renewal process of intensity $\lambda$.
    \label{cor:expected-cluster-size}
\end{proof}

\begin{remark}
    Proposition~\ref{prop:MTP} and Corollary~\ref{cor:expected-cluster-size} also hold in the restricted setting of the compact interval.
    (The value of $2\lambda + 1$ in Corollary~\ref{cor:expected-cluster-size} will of course depend on the length of the interval in the restricted setting.)
    However, the proof of Theorem~\ref{thm:weak2R} in Section~\ref{ssec:weak2R-proof} does not hold in the restricted setting.
\end{remark}

\subsection{Proof of Theorem~\ref{thm:weak2R}}
\label{ssec:weak2R-proof}
The fact that $\Xi_{\lambda}^{(t)}$ is stationary for all $0 \leq t \leq \infty$ is immediate from the assumption that $\Xi_{\lambda}^{(0)}$ is stationary and from Proposition~\ref{prop:law_C_lambda_t}.
We now show that the expected gap between clusters in $\Xi_{\lambda}^{(\infty)}$ is $
\frac{2\lambda + 1}{\lambda}$.

\begin{definition}
    A \emph{grouping} is the set of connected components in $\Xi_{\lambda}^{(t)}$ that comprises all points initially in the same connected component in $\Xi_{\lambda}^{(0)}$.
\end{definition}

Notice that the groupings in $\Xi_{\lambda}^{(t)}$ are i.i.d.~for all $0 \leq t \leq \infty$.
Let $G_i$ be the number of points in the $i$-th grouping, for $i \in \bZ$ (When $\Xi_\lambda^{(0)}$ is a Poisson point process, we have $G_1 \sim \mathrm{Geom}(e^{-\lambda})$.).
Let $N_i$ be the number of clusters in $\Xi_{\lambda}^{(\infty)}$ which are comprised of points of $G_i$.

We have the following almost sure equality due to the strong law of large numbers:,
\begin{align*}
    &\lim_{k \to \infty}\frac{G_{-k} + G_{-k + 1} + \ldots + G_k}{N_{-k} + N_{-k + 1} + \ldots + N_k} \\
    &\quad= \lim_{k \to \infty}\frac{\frac{1}{2k+1}(G_{-k} + G_{-k + 1} + \ldots + G_k)}{\frac{1}{2k+1}(N_{-k} + N_{-k + 1} + \ldots + N_k)}
    = \frac{\E G_1}{\E N_1}.
\end{align*}
Using Corollary~\ref{cor:int} and Proposition~\ref{prop:gap} we have almost surely
\begin{align*}
    &\lim_{k \to \infty}\frac{G_{-k} + G_{-k + 1} + \ldots + G_k}{N_{-k} + N_{-k + 1} + \ldots + N_k} \\
    &\quad= \lim_{k \to \infty}\E \frac{G_{-k} + G_{-k + 1} + \ldots + G_k}{N_{-k} + N_{-k + 1} + \ldots + N_k} = 2\lambda + 1.
\end{align*}
Here, the final equality follows since the expected cluster size is $2\lambda + 1$.
Thus $\E G_1 = (2\lambda + 1)\E N_1$.

Let $\Lambda(\cdot)$ denote Lebesgue measure on $\bR$.
Since $\Xi_{\lambda}^{(\infty)}$ has intensity $\lambda$, it follows that if $K_i$ is the (compact) convex hull of the initial connected components of $\Xi_{\lambda}^{(0)}$ indexed in $[-i, -i +1, \ldots, i]$, we have the almost sure limit via Proposition~\ref{prop:intensity-empirical}:
$$\lim_{i \to \infty} \frac{\sum_{j=-i}^{i}G_j}{\Lambda(K_i)} = \lim_{i \to \infty} \frac{\sum_{j=-i}^{i}\E G_j}{\Lambda(K_i)} = \lambda.$$
Using the fact that $\E N_1 = \frac{\E G_1}{2\lambda + 1}$, it follows that 
$$\lim_{i \to \infty} \frac{\sum_{j=-i}^{i}\E N_j}{\Lambda(K_i)} = \frac{\lambda}{2\lambda + 1},$$
which is the intensity of a simple point process consisting of points at the cluster locations of $\Xi_\lambda^{(\infty)}$.
The final result is an immediate consequence of Proposition~\ref{prop:intensity-gap}.

%%******

\section{Simulation Evidence and Intuition for Strong $2R$-Conjecture}
\label{sec:strong-2R-evidence}
Recall that the Strong $2R$-Conjecture states that the joint locations of the clusters converges to $\mathrm{Unif}(-1, 1) + 2\bZ$.
The uniform shift is related to the \emph{displacement} of a typical point:
\begin{definition}
    The \emph{displacement} is given by $\eta_{\lambda} := \Xi_{\lambda, 0}^{(\infty)} - \Xi_{\lambda, 0}^{(0)}.$
\end{definition}

\begin{lemma}
    Suppose that the weak limit $\lim_{\lambda \to \infty}\Xi_{\lambda}^{(\infty)}$ exists in distribution and that the gap distribution between clusters in $\Xi_{\lambda}^{(\infty)}$ converges in probability to $2$ as $\lambda \to \infty$.
    Then the displacement of the zero-indexed point converges in distribution to $S + \mathrm{Unif}(-1, 1)$, where $S$ is a symmetric random variable (for all Borel-measurable sets $A \subseteq \bR$, we have $\pr(S \in A) = \pr(S \in -A)$).
\end{lemma}
\begin{proof}
    Recalling that $\Xi_{\lambda}^{(0)}$ is stationary for all $\lambda$, it follows from the dynamics that $\Xi_{\lambda}^{(\infty)}$ is also stationary for all $\lambda$.
    Under the assumed convergence of $\Xi_{\lambda}^{(\infty)}$, any limit must also be stationary.
    The distributional convergence of the displacement is a result of the assumption that the weak limit $\lim_{\lambda \to \infty}\Xi_{\lambda}^{(\infty)}$ exists in distribution.

    Since the gap distribution between clusters is assumed to converge in probability to $2$, it follows from the definition of stationary point processes that the nearest cluster to the right of the origin converges in distribution to $\mathrm{Unif}(0, 2)$.
    As a result, the location of the nearest cluster to the origin converges in distribution to $\mathrm{Unif}(-1, 1)$.
    The result follows from the symmetry that for all $i \in \bZ$, and $0 \leq t \leq \infty$, 
    $
    \Xi_{\lambda}^{(t)}(i) - \Xi_{\lambda}^{(t)}(0) \equald \Xi_{\lambda}^{(t)}(0) - \Xi_{\lambda}^{(t)}(-i).
    $
\end{proof}

In particular, the Strong $2R$-Conjecture states that the random variable $S = 0$ a.s.~in the previous result.
The displacement behavior can be verified through simulation.
Despite the fact that our simulation is of an atypically small connected component lying in an interval of length $30$ (see Remark~\eqref{rem:interval-atypical}), we nevertheless observe that the distribution of the displacement increases almost linearly within the interval $[-1, 1]$, which matches the Strong $2R$-Conjecture.

\begin{figure}
    \centering
    \includegraphics[width=0.7\linewidth]{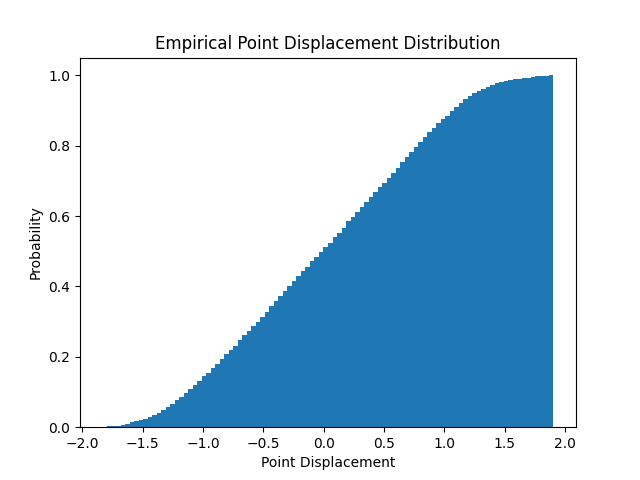}
    \caption{The empirical distribution of point displacements in a simulation with initial interval length $30$ and an initial Poisson point process of intensity $250$.}
    \label{fig:emp-displacement}
\end{figure}

%%******

\section{Future Work}
\label{sec:future-work}
In this paper, we extend the \hk model to a point-process-valued state space, and use various symmetries in that space to partially resolve the $2R$-Conjecture with elementary methods.
There are several interesting and important directions for future work based on the techniques presented in this paper.

\begin{enumerate}[left=0pt, label={\alph*}.]
    \item \textbf{Proving the Strong $2R$-Conjecture: } First and foremost, the Strong $2R$-Conjecture remains unproven, although the methodology and framework used in this paper allow us to give the first precise statement of this conjecture.
    Although our use of techniques based on the mass transport principle and the strong law of large numbers does not readily facilitate an understanding beyond expectations, the interplay between point process and random geometric graph descriptions of the dynamics may be a fruitful avenue for future work.
    Unlike~\cite{dragovic2020poisson, khaniha2025hierarchical, angel2023tale}, the limit in this problem has been identified, and thus the \hk model and related models may be good candidates through which one can develop appropriate tools.
    
    \item \textbf{Model Generalizations -- Convergence and Fixed Point Characterization}
    Typical methods for proving convergence of generalizations of the \hk models have focused on the restricted setting of finitely many points located within a compact set.
    Typically, convergence has been studied using Lyapunov methods, which do not aid in characterizing the fixed point.
    Based on the current status of the literature, we currently lack tools to characterize the structure of the fixed points of generalizations of the \hk model.
    
    The extended setting in this model could provide a unified framework through which one can simultaneously understand the convergence to and structure of the fixed point.
    
    \begin{itemize}[left=0pt]
        \item \textbf{Inhomogeneous \hk Models:} In the Inhomogeneous \hk model, each agent uses its own (i.i.d.) radius for the ball in which it averages over other agents.
        Chazelle~\cite{chazelle2015diffusive} has demonstrated that those models can exhibit chaotic behavior within a compact set from certain initial conditions.
        The existence of a weak limit remains an open problem.
        We believe that when the initial condition is a Poisson point process, a weak limit should exist and the dynamics should not exhibit chaotic behavior.
        
        \item \textbf{Multi-Dimensional \hk Models:} The convergence of the multidimensional \hk model has been shown for finitely many points.
        However, in more than one dimension, percolation occurs, which means that the existence of the weak limit in the extended setting is no longer a simple consequence of the convergence of the finite model.
        Even the existence of a weak limit in our extended model is an open problem.
        While our mass transport argument extends to multiple dimensions, the more general result is not immediate, as it is not obvious whether or not a percolated state can persist under the \hk dynamics.
        The persistence of a percolated state at all finite times does not preclude a grid structure for $\lim_{\lambda}\lim_{t}\Xi_\lambda^{(t)}$.
        As a result, the generalization to multiple dimensions may require different techniques from those used in this paper.
    \end{itemize}
\end{enumerate}

%%******

\bibliographystyle{plain}
\bibliography{references}

\begin{thebibliography}{10}

\bibitem{aggarwal2003framework}
Charu~C Aggarwal, S~Yu Philip, Jiawei Han, and Jianyong Wang.
\newblock A framework for clustering evolving data streams.
\newblock In {\em Proceedings 2003 VLDB Conference}, pages 81--92. Elsevier, 2003.

\bibitem{aldous2007processes}
David Aldous and Russell Lyons.
\newblock Processes on unimodular random networks.
\newblock 2007.

\bibitem{angel2023tale}
Omer Angel, Gourab Ray, and Yinon Spinka.
\newblock A tale of two balloons.
\newblock {\em Probability Theory and Related Fields}, 185(3):815--837, 2023.

\bibitem{bernardo2024bounded}
Carmela Bernardo, Claudio Altafini, Anton Proskurnikov, and Francesco Vasca.
\newblock Bounded confidence opinion dynamics: A survey.
\newblock {\em Automatica}, 159:111302, 2024.

\bibitem{bhattacharyya2013convergence}
Arnab Bhattacharyya, Mark Braverman, Bernard Chazelle, and Huy~L. Nguyen.
\newblock On the convergence of the hegselmann-krause system.
\newblock In {\em Proceedings of the 4th Conference on Innovations in Theoretical Computer Science}, ITCS '13, page 61–66, New York, NY, USA, 2013. Association for Computing Machinery.

\bibitem{blaszczyszyn2017lecture}
Bart{\l}omiej B{\l}aszczyszyn.
\newblock Lecture notes on random geometric models---random graphs, point processes and stochastic geometry.
\newblock 2017.

\bibitem{blondel20072r}
Vincent~D Blondel, Julien~M Hendrickx, and John~N Tsitsiklis.
\newblock On the $2r$ conjecture for multi-agent systems.
\newblock In {\em 2007 European Control Conference (ECC)}, pages 874--881. IEEE, 2007.

\bibitem{blondel2010continuous}
Vincent~D Blondel, Julien~M Hendrickx, and John~N Tsitsiklis.
\newblock Continuous-time average-preserving opinion dynamics with opinion-dependent communications.
\newblock {\em SIAM Journal on Control and Optimization}, 48(8):5214--5240, 2010.

\bibitem{bullo2009distributed}
Francesco Bullo, Jorge Cort{\'e}s, and Sonia Martinez.
\newblock {\em Distributed control of robotic networks: {A} mathematical approach to motion coordination algorithms}.
\newblock Princeton University Press, 2009.

\bibitem{chazelle2015diffusive}
Bernard Chazelle.
\newblock Diffusive influence systems.
\newblock {\em SIAM Journal on Computing}, 44(5):1403--1442, 2015.

\bibitem{dragovic2020poisson}
Natasa Dragovic.
\newblock {\em On the Poisson Follower Model}.
\newblock PhD thesis, The University of Texas at Austin, 2020.

\bibitem{gwynne2018invariance}
Ewain Gwynne, Jason Miller, and Scott Sheffield.
\newblock An invariance principle for ergodic scale-free random environments.
\newblock {\em arXiv preprint arXiv:1807.07515}, 2018.

\bibitem{hegselmann2002opinion}
Rainer Hegselmann and Ulrich Krause.
\newblock Opinion dynamics and bounded confidence: {M}odels, analysis and simulation.
\newblock 2002.

\bibitem{khaniha2025hierarchical}
Sayeh Khaniha and Fran{\c{c}}ois Baccelli.
\newblock Hierarchical clustering algorithms on poisson and cox point processes.
\newblock {\em arXiv preprint arXiv:2503.18555}, 2025.

\bibitem{wang2017noisy}
Chu Wang, Qianxiao Li, Weinan E, and Bernard Chazelle.
\newblock Noisy hegselmann-krause systems: Phase transition and the 2 r-conjecture.
\newblock {\em Journal of Statistical Physics}, 166:1209--1225, 2017.

\bibitem{wedin2022mathematics}
Edvin Wedin.
\newblock {\em On the mathematics of the one-dimensional Hegselmann-Krause model}.
\newblock 2022.

\end{thebibliography}

\end{document}